\documentclass[12pt]{article}
\usepackage[utf8]{inputenc}
\usepackage{amsthm}
\usepackage{amsfonts}
\usepackage{amsmath}
\usepackage{colonequals}
\usepackage{txfonts}
\usepackage{url}
\newtheorem{theorem}{Theorem}

\newtheorem{lemma}[theorem]{Lemma}
\newtheorem{corollary}[theorem]{Corollary}

\newcommand{\iocp}{\mathrm{iocp}}
\newcommand{\Prob}{\mathrm{Prob}}
\renewcommand{\SS}{\mathcal{S}}
\newcommand{\XX}{\mathcal{X}}

\newcommand{\E}{{\mathbb{E}}}

\title{Induced Odd Cycle Packing Number, Independent Sets, and Chromatic Number\thanks{An extended abstract of this paper was published in proceedings of EUROCOMB 2019}}

\author{Zden\v{e}k Dvo\v{r}\'ak\thanks{Charles University, Prague, Czech Republic.  E-mail: {\tt rakdver@iuuk.mff.cuni.cz}.
Supported by project 17-04611S (Ramsey-like aspects of graph coloring) of Czech Science Foundation.}\and
Jakub Pek\'arek\thanks{Charles University, Prague, Czech Republic.  E-mail: {\tt pekarej@iuuk.mff.cuni.cz}.
Supported by GAUK grant 118119 (Algorithms for graphs with restrictions on cycles).}}

\begin{document}
\maketitle

\begin{abstract}
The \emph{induced odd cycle packing number} $\iocp(G)$ of a graph $G$ is the maximum integer $k$ such that $G$
contains an induced subgraph consisting of $k$ pairwise vertex-disjoint odd cycles.
Motivated by applications to geometric graphs, Bonamy et al.~\cite{indoc} proved that
graphs of bounded induced odd cycle packing number, bounded VC dimension, and linear independence number
admit a randomized EPTAS for the independence number.  We show that the assumption of bounded VC dimension
is not necessary, exhibiting a randomized algorithm that for any integers $k\ge 0$ and $t\ge 1$
and any $n$-vertex graph $G$ of induced odd cycle packing number at most $k$ returns in time $O_{k,t}(n^{k+4})$ an independent
set of $G$ whose size is at least $\alpha(G)-n/t$ with high probability. In addition, we present $\chi$-boundedness results
for graphs with bounded odd cycle packing number, and use them to design a QPTAS for the independence number
only assuming bounded induced odd cycle packing number.
\end{abstract}

\section{Introduction}

The graph classes defined by forbidden cycles or induced cycles of certain
lengths figure prominently in the structural graph theory, motivated in
particular by the Strong Perfect Graph Theorem~\cite{sperft}, which shows that
perfect graphs are characterized by forbidden odd holes and their complements. 
The most interesting graph parameters in the context of these graph classes are
the chromatic number, the independence number, and the clique number: while they are NP-hard to
approximate within any fixed precision~\cite{colnonap} in general graphs, using semidefinite programming
they can be determined in polynomial time for perfect graphs~\cite{GLS}. 

Perfect graphs also motivate the concept of $\chi$-boundedness. A class of graphs is \emph{$\chi$-bounded}
if the chromatic number of the graphs from the class can be bounded by a function of the clique number
(of course, there is no such function in general, due to numerous known constructions of triangle-free
graphs of arbitrarily large chromatic number).  The notion of $\chi$-boundedness was introduced by Gyarf\'{a}s~\cite{gyarfas1987},
who also proposed a number of influential questions on this topic. As an
example, he conjectured that graphs without odd holes are $\chi$-bounded; this
conjecture was only recently confirmed by Scott and Seymour~\cite{SS}. In a similar
vein, Bonamy, Charbit, and Thomass\'{e}~\cite{BCT} showed that graph classes that forbid
induced cycles of length that is a multiple of 3 have bounded chromatic number, inspiring further study of connection between induced cycles of specific lengths and chromatic number. 

The aforementioned ideas commonly appear in the context of geometrically defined graph classes.
The main motivation for our work comes from recent algorithmic results exploiting the following property
\begin{itemize}
\item[($\star$)] no induced subgraph is the disjoint union of two odd cycles,
\end{itemize}
which holds in certain geometric graph classes.  Specifically, Bonnet et al.~\cite{BGKRS} proved that
the complements of intersection graphs of disks in the plane have the property ($\star$), and gave a QPTAS and
an exact subexponential-time algorithm for the independence number of graphs satisfying the property ($\star$).
In combination, this gives a QPTAS and an exact subexponential-time algorithm for the clique number of
intersection graphs of disks\footnote{Bonnet et al.~\cite{BGKRS} only explicitly give the algorithms
for intersection graphs of disks, but the inspection of their algorithms shows that they only use the property ($\star$),
and no other properties specific to intersection graphs of disks.}.
Building upon these results, Bonamy et al.~\cite{indoc} proved that the complements of the intersection graphs of unit balls
in 3-dimensional space also have the property ($\star$).  Moreover, they gave a randomized algorithm
to approximate the clique number arbitrarily well in polynomial time (a randomized efficient polynomial-time approximation scheme)
for intersection graphs of disks and of unit balls; however, in addition to the property ($\star$),
they use further properties derived from the geometry of the problems, namely bounded VC dimension and
a linear lower bound on the independence number.

More generally, let the \emph{induced odd cycle packing number}, $\iocp(G)$ of a graph $G$
be the maximum integer $k$ such that $G$ contains an induced subgraph
consisting of $k$ pairwise vertex-disjoint odd cycles (hence, $G$ has the property ($\star$) if and only
of $\iocp(G)\le 1$). Note that whenever two cycles are connected by an edge, their vertices do not induce such a subgraph;
furthermore, unlike the perfect graph case, we work with all odd cycles, not just with
odd holes (i.e., $\iocp$ takes into account also triangles).
Bonamy et al.~\cite{indoc} proved the following.

\begin{theorem}[Bonamy et al.~\cite{indoc}]
There exists a randomized algorithm taking as an input integers $k,b,c,a\ge 1$ and
graph $G$ of VC dimension at most $c$ such that $\iocp(G) \leq k$ and $\alpha(G)\ge |V(G)|/b$, and in time
$O_{k,b,c,a}(|V(G)|^2)$ returns an independent set of $G$ of size at least $(1-1/a)\alpha(G)$.
\end{theorem}

Later (private communication), they proved that one can remove the assumption that the VC dimension is bounded, at the expense
of making the exponent in the time complexity depend on the desired precision,
i.e., obtaining a PTAS with time complexity $O(|V(G)|^{f(k,b,a)})$ for some function $f$ rather than an EPTAS.
As our first result, we show that it is not necessary to make this sacrifice, obtaining an EPTAS without
the assumption of bounded VC dimension. We present the improved argument in section \ref{secEPTAS}.

In section \ref{secChiBound} we focus our attention on coloring. We show that the class of graphs $G$ such that $\iocp(G) \leq k$ are
$\chi$-bounded by a function polynomial in the clique number (with the degree of the polynomial depending linearly on $k$).
Furthermore, our proof of this fact can be turned into an algorithm running in polynomial time for fixed maximum clique size.
In section \ref{secLowerBound} we follow by a lower bound on the $\chi$-bounding function for the case when $k=1$.
Finally, in section~\ref{secQPTAS}, we show how these results can be combined to obtain a QPTAS for
the maximum independent set in graphs with bounded induced odd cycle packing number (generalizing a result of Bonnet et al.~\cite{BGKRS},
who gave such a QPTAS for graphs with induced odd cycle packing number at most one).

Let us remark that Bock et al.~\cite{bock} considered a non-induced version of the notion,
and in particular give approximation algorithms for the independence number of graphs with bounded
number of pairwise-disjoint odd cycles. However, this setting is quite distinct from ours, since
e.g. a clique $K_n$ contains $\lfloor n/3\rfloor$ disjoint odd cycles, but $\iocp(K_n)\le 1$.

\section{Graphs of large odd girth}

Our EPTAS is structured as follows. First, we deal with the case of high odd-girth as a base case.
Then we investigate graph that cannot be reduced to this base case without distorting the solution too
much due to many disjoint short odd cycles.  We show that in such case, it is
possible to either reduce $\iocp$ or (with high probability) the number of vertices of the input graph without
distorting the solution too much, producing a recursive algorithm.

In this section, we deal with the high odd girth case.  Note that this is only a mild variation on what
already appeared in Bonamy et al.~\cite{indoc}; we describe it here in detail, as they do not state
the result separately, only as a part of a longer argument.  We also give the result in larger generality
than we actually need, in a weighted setting, as this is quite natural (unfortunately, the rest of our
argument only works in unweighted setting).

The \emph{odd girth} of a graph $G$ is the length of shortest odd cycle appearing in $G$.
We need an observation on neighborhoods of such shortest odd cycles. Essentially, it is possible to remove a small ``wedge''
from the close neighborhood of the cycle (destroying the odd cycle in the process) so that the remainder of the neighborhood is bipartite. 

For a set $Y\subseteq V(G)$, let $N(Y)$ denote the set of vertices at distance exactly $1$ from $Y$ (vertices from $V(G)\setminus Y$ with a neighbor in $Y$).
For a positive integer $r$, let $N_r[Y]$ denote the set of vertices at distance at most $r$ from $Y$; note that $Y\subseteq N_r[Y]$.
Let $N[Y]=N_1[Y]$.  For a vertex $v\in V(G)$, let $N(v)=N(\{v\})$, $N_r[v]=N_r[\{v\}]$, and $N[v]=N[\{v\}]$.

\begin{lemma}\label{lemma-short-nbr}
Let $g$ be an odd integer, let $G$ be a graph of odd girth at least $g$, and let $C$ be a shortest odd cycle in $G$.
Let $t\le (g-1)/2$ be a non-negative integer, and suppose that every vertex of $G$ is at distance at most $t$ from $C$.
Let $z$ be a vertex of $C$, let $A = N_t[z] \cap C$,
and let $R = N_t[A]$.  Then $G-R$ is bipartite.
\end{lemma}
\begin{proof}
Note that $C$ is geodesic in $G$, i.e., the distance between any two vertices of $C$ is the same in $C$ as in $G$
(otherwise, $G$ would contain a path $Q$ between two vertices $x$ and $y$ of $C$ shorter than the distance between $x$ and $y$ in $C$,
and $C\cup Q$ would contain an odd cycle shorter than $C$), and in particular $C$ is an induced cycle.

Let $F$ be a forest of shortest paths from vertices in $G$ to $V(C)$, and for each $v\in V(G)$,
let $f(v)$ denote the vertex in which the component of $F$ containing $v$ intersects $C$.
Note that for each $v\in V(G)$, the distance in $F$ from $v$ to $f(v)$ is at most $t$.
Observe that $(F\cup C)-z$ is also a forest, and thus it has a proper
$2$-coloring $\psi$.  We claim that the restriction of $\psi$ is a proper $2$-coloring of $G-R$.
Suppose for a contradiction there exists an edge $uv\in E(G-R)$ with $\psi(u)=\psi(v)$.
Since $u,v\not\in R$, we have $f(u)\neq z\neq f(v)$, and thus there exists a unique path $P$ between $u$ and $v$ in $(F\cup C)-z$.
Since $\psi(u)=\psi(v)$, the cycle $P+uv$ has odd length, and since $C$ is a shortest odd cycle in $G$, we have
$|E(P)|+1\ge |C|\ge g\ge 2t+1$.
In particular, $P$ contains both $f(u)$ and $f(v)$. Let $P'$ be the subpath of $P$ between $f(u)$ and $f(v)$; we have
$|E(P')|\ge |E(P)|-2t\ge |C|-2t-1$.  Since $P'$ is a subpath of the path $C-z$ and $|E(C-z)|=|C|-2$,
we can by symmetry assume that the distance between $f(u)$ and $z$ is at most $t$, and thus $f(u)\in A$.
But then $u\in R$, which is a contradiction.
\end{proof}

Note that if $|C|\gg t$, then the above lemma can be used to obtain many disjoint
sets whose removal makes the graph bipartite, and thus one of them must contain
only a small fraction of vertices, or, in a weighted setting, contain only a small fraction
of the total weight.  We now use this observation to decrease the induced odd cycle packing number
without decreasing the weight of the heaviest independent set too much.
Given an assignment $w:V(G)\to\mathbb{Z}^+$ of weights to vertices of $G$, let for each set $X\subseteq V(G)$
define $w(X)=\sum_{v\in X} w(v)$ and let $\alpha_w(G)$ be the maximum of $w(X)$ over all independent sets $X$ in $G$.

\begin{lemma}\label{lemma-decrio}
There exists an algorithm that, given an integer $b\ge 1$, an $n$-vertex non-bipartite graph $G$ of odd girth at least $2b(8b-3)$,
and an assignment $w:V(G)\to\mathbb{Z}^+$ of weights to vertices, returns in time $O(bn^2+n^3)$ induced subgraphs $G_1$, \ldots, $G_{2b}$
of $G$ such that $\iocp(G_i)\le \iocp(G)-1$ for $i\in\{1,\ldots, 2b\}$, and
\begin{align*}
\max\{\alpha_w(G_i):i\in\{1,\ldots,2b\}\}&\ge (1-1/b)\alpha_w(G)\text{, and}\\
\max\{w(V(G_i)):i\in\{1,\ldots,2b\}\}&\ge (1-1/b)w(V(G)).
\end{align*}
\end{lemma}
\begin{proof}
We find (in time $O(n^3)$ by BFS from each vertex) a shortest odd cycle $C$ in $G$,
necessarily of length at least $2b(8b-3)$.  The argument from the proof of Lemma~\ref{lemma-short-nbr} shows that
$C$ is geodesic in $G$.  Let $Z=\{z_1,\ldots, z_{2b}\}$ be a set of vertices of $C$ at distance at least $8b-3$
from one another and for $i\in\{1,\ldots,2b\}$, let $R_i=N_{4b-2}[z_i]$ denote the set of vertices of $G$ at distance at most $4b-2$ from $z_i$
(the choice of $z_1$, \ldots, $z_{2b}$ implies these sets are pairwise disjoint).
Let $L_i$ denote the set of vertices of $G$ at distance exactly $i$ from $C$.  Let $G_i=G-L_i-R_i$.

Consider any $i\in\{1,\ldots,2b\}$.  We claim that $\iocp(G_i)\le \iocp(G)-1$.  Indeed, let $G_1$ be the subgraph of $G$ induced
by vertices at distance less than $i$ from $C$
and $G_2$ the subgraph induced by vertices at distance greater than $i$ from $C$, so that $G-L_i$ is the disjoint union of $G_1$ and $G_2$.
Since $G_1$ contains the odd cycle $C$ whose vertices have no neighbors in $G_2$, we have $\iocp(G_2)\le \iocp(G-L_i)-1\le \iocp(G)-1$.
Furthermore, by Lemma~\ref{lemma-short-nbr} (for $t = 4b-2$), the graph $G_1-R_i$ is bipartite, and thus
$\iocp(G_i)=\iocp(G_1-R_i)+\iocp(G_2-R_i)\le \iocp(G)-1$.

Consider a heaviest independent set $I$ in $G$.  Since the sets $L_1$, \ldots, $L_{2b}$
are pairwise disjoint, and the sets $R_1$, \ldots, $R_{2b}$ are pairwise disjoint as well, we have
$$\sum_{i=1}^{2b} w(I\cap (L_i\cup R_i))\le 2w(I),$$ and thus there exists $i\in \{1,\ldots, 2b\}$ such that $w(I\cap (L_i\cup R_i))\le w(I)/b$.
Hence, $\alpha_w(G_i)\ge (1-1/b) \alpha_w(G)$.  Similarly,
$$\sum_{j=1}^{2b} w(L_j\cup R_j)\le 2w(V(G)),$$
and thus there exists $j\in\{1,\ldots,2b\}$ such that $w(L_j\cup R_j)\le w(V(G))/b$ and $w(V(G_j))\ge (1-1/b)w(V(G))$.
\end{proof}

Let us remark on a lower bound on $\alpha_w(G)$ in graphs satisfying the assumptions of Lemma~\ref{lemma-decrio}.
\begin{corollary}\label{cor-lower}
Let $k\ge 0$ and $b\ge 1$ be integers and let $G$ be a graph of induced odd cycle packing number at most $k$ and odd girth at least $2b(8b-3)$.
For any assignment $w:V(G)\to\mathbb{Z}^+$ of weights to vertices, we have
$$\alpha_w(G)\ge \frac{(1-1/b)^k w(V(G))}{2}\ge \frac{(1-k/b)w(V(G))}{2}.$$
\end{corollary}
\begin{proof}
We prove the claim by induction on $k$.  If $G$ is bipartite, then one of the parts of bipartition of $G$ has weight at least $w(V(G))/2$.
Hence, we can assume that $G$ is not bipartite, and thus $k\ge 1$.  By Lemma~\ref{lemma-decrio}, there exists an induced subgraph $G'$
of $G$ of induced odd cycle packing number at most $k-1$ such that $w(V(G'))\ge (1-1/b)w(V(G))$, and by the induction hypothesis,
$$\alpha_w(G)\ge \alpha_w(G')\ge \frac{(1-1/b)^{k-1} w(V(G'))}{2}\ge \frac{(1-1/b)^k w(V(G))}{2}.$$
\end{proof}

Similarly, iterating Lemma~\ref{lemma-decrio}, we obtain an approximation for the maximum-weight independent set
in graphs of bounded induced odd cycle packing number.

\begin{lemma}\label{lemma-noshort}
There exists an algorithm that, given integers $k\ge 0$ and $b\ge 1$, an $n$-vertex graph $G$
of induced odd cycle packing number at most $k$ and odd girth at least $2b(8b-3)$,
and an assignment $w:V(G)\to\mathbb{Z}^+$ of weights to vertices,
returns in time $O((2b)^kn^3)$ an independent set $X\subseteq V(G)$ such that $$w(X)\ge (1-1/b)^k\alpha_w(G)\ge (1-k/b)\alpha_w(G).$$
\end{lemma}
\begin{proof}
We prove the claim by induction on $k$.  If $G$ is bipartite, then we can find an independent set in $G$ of
largest weight via a maximum flow algorithm in time $O(n^3)$;
and considering the heavier of the two color classes of $G$, we have $\alpha_w(G)\ge w(V(G))/2$.
Hence, we can assume that $G$ is not bipartite, and in particular $k\ge 1$.

Let $G_1$, \ldots, $G_{2b}$ be the induced subgraphs of $G$ of induced odd cycle packing number at most $k-1$ obtained
using the algorithm from Lemma~\ref{lemma-decrio}.  We now recurse on $G_1$, \ldots, $G_{2b}$ (with $k$ replaced by $k-1$)
obtaining by the induction hypothesis in time $2b\cdot O((2b)^{k-1}n^3)=O((2b)^kn^3)$ independent sets $X_1$, \ldots, $X_{2b}$
such that $w(X_i)\ge (1-1/b)^{k-1}\alpha_w(G_i)$ for each $i\in \{1,\ldots,2b\}$.  We return the heaviest of these independent
sets; we have
\begin{align*}
\max\{w(X_i):i\in\{1,\ldots,2b\}\}&\ge (1-1/b)^{k-1}\max\{\alpha_w(G_i):i\in\{1,\ldots,2b\}\}\\
&\ge (1-1/b)^{k-1}\cdot (1-1/b)\alpha_w(G)=(1-1/b)^k\alpha_w(G).
\end{align*}
\end{proof}

\section{EPTAS assuming linear independence number}\label{secEPTAS}

Let us now move on to the general case of graphs with bounded induced odd cycle packing number.
From now on, we work in unweighted setting.
In order to make use of Lemma~\ref{lemma-noshort}, it is necessary to destroy all short odd cycles in
the input graph graph. We show that unless simply deleting a maximum packing of such cycles erases
only a small portion of the graph (and hence yields a useful approximation) there exist useful structures. 
Either there exists a cycle which can be used to decrease $\iocp$ without decreasing the size of the largest independent set too much or there are many high-degree vertices which can be used to guess part of the solution while significantly reducing
the remainder of the graph. 

For a set $\SS$ of graphs, an \emph{$\SS$-packing} in a graph $G$ is a set $\XX$ of pairwise vertex-disjoint induced subgraphs of $G$, each isomorphic to a graph belonging to $\SS$
(note that we allow edges between members of $\XX$, unlike in the definition of $\iocp$).  Let $V(\XX)=\bigcup_{X\in \XX} V(X)$.
For an integer $g\ge 3$, let $\SS_g$ denote the set of all odd cycles of length less than $g$.
A maximal $\SS_g$-packing $\XX$ in an $n$-vertex graph $G$ can be found in time $O(n^4)$ by repeatedly finding a shortest induced odd cycle and deleting it from $G$;
observe that $G-V(\XX)$ has odd girth at least $g$.  
Lemma~\ref{lemma-noshort} is used to deal with graphs without odd cycles of length less than $g$; so, it remains to handle the graphs containing an $\SS_g$-packing
covering a large fraction of the vertices.

\begin{lemma}\label{lemma-decomp}
There exists an algorithm that, for input integers $k,p\ge 1$ (with $p\ge k$) and an $n$-vertex graph $G$ of induced odd cycle packing number at most $k$, returns in
time $O(n^4+(4p)^kn^3)$ an independent set $I$ and induced odd cycles $C_1$, \ldots, $C_m$ (for some $m\le n$) in $G$, such that at least one of the following claims holds:
\begin{itemize}
\item[(a)] $|I|\ge (1-k/p)\alpha(G)$, or
\item[(b)] there exists $i\in \{1,\ldots,m\}$ such that $\alpha(G-N(C_i))\ge(1-1/p)\alpha(G)$, or
\item[(c)] there are at least $\tfrac{k}{81920p^6}\alpha(G)$ vertices $v\in V(G)$ of degree at least $\tfrac{k}{81920p^6}n$ such that $\alpha(G-N(v)) = \alpha(G)$.
\end{itemize}
\end{lemma}
\begin{proof}
Let $g=4p(16p-3)$.  Let $\XX$ be a maximal $\SS_g$-packing in $G$.  Let $I$ be an independent set in $G-V(\XX)$ found using the algorithm from Lemma~\ref{lemma-noshort} with $b=2p$.
Let $\XX=\{C_1,\ldots, C_m\}$ and for $i\in \{1,\ldots,m\}$, let $H_i=G-N(C_i)$. 

Suppose first that $|V(\XX)|\le \tfrac{k}{10p}n$.  Since $\tfrac{k}{10p}\le 0.1$, Corollary~\ref{cor-lower} implies
$$\alpha(G)\ge\alpha(G-V(\XX))\ge (1-k/b)|V(G-V(\XX))|/2\ge 0.45(1-k/b)n\ge 0.225n,$$
and thus $|V(\XX)|\le \tfrac{k}{10p}n\le \tfrac{k}{2p}\alpha(G)$.
Consequently, $\alpha(G-V(\XX))\ge \alpha(G)-|V(\XX)|\ge \bigl(1-\tfrac{k}{2p}\bigr)\alpha(G)$, and
thus the set $I$ returned by the algorithm from Lemma~\ref{lemma-noshort} satisfies
$$|I|\ge \bigl(1-\tfrac{k}{2p}\bigr)\alpha(G-V(\XX))\ge (1-k/p)\alpha(G),$$
implying that (a) holds.

Hence, we can assume $|V(\XX)|>\tfrac{k}{10p}n$.  Let $J$ be a largest independent set in $G$.  If there exists $i\in \{1,\ldots,m\}$ such that $|N(C_i)\cap J|\le |J|/p$,
then (b) holds.  Hence, assume that $|N(C_i)\cap J|>|J|/p$ for every $i\in \{1,\ldots,m\}$; and consequently, for each such $i$, there exists $b_i\in V(C_i)$ such that
$|N(b_i)\cap J|>\tfrac{|J|}{pg}$.  Let $B=\{b_1,\ldots,b_m\}$ and note that $m\ge |V(\XX)|/g>\tfrac{k}{10pg}n$.  By double-counting the number of edges of $G$ between $B$ and $J$,
we have
\begin{align*}
\sum_{v\in J} \deg v&\ge\sum_{v\in J} |N(v)\cap B|=\sum_{b\in B} |N(b)\cap J|\\
&>\frac{m|J|}{pg}>\frac{k}{10p^2g^2}|J|n.
\end{align*}
Let $J'$ consist of vertices of $J$ of degree greater than $\tfrac{k}{20p^2g^2}n\ge \tfrac{k}{81920p^6}n$.
Since $J$ is a largest independent set in $G$, for every $v\in J$ we have $\alpha(G-N(v))=\alpha(G)$.  Note that
\begin{align*}
\sum_{v\in J'} \deg v&= \sum_{v\in J}  \deg v - \sum_{v\in J\setminus J'} \deg v\\
&> \frac{k}{10p^2g^2}|J|n-\tfrac{k}{20p^2g^2}n\cdot |J\setminus J'|\ge\frac{k}{20p^2g^2}|J|n,
\end{align*}
and since each vertex has degree less than $n$, we have $|J'|>\tfrac{k}{20p^2g^2}|J|=\tfrac{k}{20p^2g^2}\alpha(G)\ge \tfrac{k}{81920p^6}\alpha(G)$.
\end{proof}

The possible outcomes of Lemma \ref{lemma-decomp} offer a natural recursive
approximation algorithm. Supposing a suitable setting of the parameters,
if the case (a) holds, then the algorithm provides a good independent set. If (b)
holds, then we may restrict the problem to an induced subgraph of lower $\iocp$
while preserving the solution well enough. If (c) holds, then we are guaranteed
many vertices of large degree such that deleting their neighbors does not
affect the solution.  A random sampling of vertices with high degree
provides a good probability of hitting one of these vertices.
We show that this approach yields a randomized EPTAS under assumption that $\alpha(G)=\Omega(|V(G)|)$.

\begin{lemma}\label{lemma-main}
There exists a randomized algorithm that, for input integers $k\ge 0$ and $t\ge 1$ and an $n$-vertex graph $G$ of induced odd cycle packing number at most $k$, returns in
time $O_{k,t}(n^4)$ an independent set of $G$ whose size is at least $\alpha(G)-n/t$ with probability at least $\Omega_{k,t}(n^{-k})$.
\end{lemma}
\begin{proof}
Let $p=kt$, and when $k\neq 0$, let $q=\tfrac{1}{81920p^6}$ and $d\ge 1$ be the smallest integer such that $(1-q)^d<1/p$.
Let us now describe a recursive procedure that, applied to an induced subgraph $G'$ of $G$ and a non-negative integer $k'\le k$ such that $\iocp(G')\le k'$,
returns an independent set of $G'$; as we will show later, this set has size at least $\alpha(G')-nk'/p$ with probability $\Omega_{k,t}(n^{-k'})$.

If $|V(G')|\le nk'/p$, then we return an empty set (or any other independent set in $G'$).
If $k'=0$, then $G'$ is bipartite and we can find a largest independent set in $G'$ via a maximum flow algorithm in time $O(n^3)$.
Hence, suppose that $k'\ge 1$.  Apply the algorithm from Lemma~\ref{lemma-decomp} to $G'$ (using $k'$ as $k$) to obtain an independent set $I$ and cycles $C_1,\dots,C_m$.
Now, we at random perform one of the following actions, each with probability $1/3$:
\begin{itemize}
\item[(a)] Return the set $I$.
\item[(b)] Choose $i\in \{1,\ldots,m\}$ uniformly at random.  Note that the induced subgraph $H_i = G' - N[C_i]$ satisfies
$\iocp(H_i)\le k'-1$.  Let $I'$ be an independent set in $H_i$ obtained by a recursive call for $H_i$ with $k'$
replaced by $k'-1$.  Return the union of $I'$ with an independent set of $C_i$ of size $\alpha(C_i)=\lceil |C_i|/2\rceil$.
\item[(c)] Choose a vertex $u\in V(G')$ of degree at least $k'q|V(G')|$ uniformly at random (if no such vertex
exists, return an empty set, instead).  Return the independent set obtained by the recursive call for $G'-N(u)$, with the same $k'$.
\end{itemize}
Let us analyze the running time of this procedure when applied to $G$ with $k'=k$.  Note that at each level of the recursion,
we only perform one recursive call, and either $k'$ decreases by one, or the number of vertices decreases by the factor
smaller or equal to $1-\tfrac{1}{81920p^6}$.  Moreover, the recursion stops when the number of vertices is at most $n/p$ (or earlier).
Hence, the total depth of the recursion is at most $k+d$, and since $d$ only depends on $k$ and $t$, the running time
of the algorithm is $O_{k,t}(n^4)$.

Let $d(G')$ be the smallest non-negative integer such that $(1-q)^{d(G')}|V(G')|<n/p$; in particular, $d(G)=d$.
Let $a(G',k')=(3n)^{-k'}\bigl(\tfrac{q}{3p}\bigr)^{d(G')}$.
Let us now show that the set returned by the algorithm for $G'$ and $k'$ has size at least $\alpha(G')-nk'/p$ with probability
at least $a(G',k')$.  Note that when applied to $G'=G$ and $k'=k$, this implies the algorithm
returns an independent set of $G$ of size at least $\alpha(G)-n/t$ with probability $\Omega_{k,t}(n^{-k})$, as required.

We prove the claim by induction on $k'+|V(G')|$.  If $k'=0$, then we return an optimal independent set, and if $\alpha(G')\le nk'/p$, then the claim
is trivial.  Hence, we can assume that $k'>0$ and $\alpha(G')>nk'/p$.  Let us now distinguish which outcome of
Lemma~\ref{lemma-decomp} holds.
\begin{itemize}
\item If (a) holds, then with probability $1/3\ge a(G',1)\ge a(G',k')$, we return the independent set $I$, which has
size at least $(1-k'/p)\alpha(G')\ge \alpha(G')-nk'/p$.
\item If (b) holds, then with probability $\tfrac{1}{3m}\ge \tfrac{1}{3n}$,
the algorithm takes the branch (b) and chooses $i\in\{1,\ldots,m\}$ such that $\alpha(G'-N(C_i))\ge(1-1/p)\alpha(G')$.
By the induction hypothesis, with probability $a(H_i,k'-1)\ge a(G',k'-1)$, the independent set $I'$ returned by the
recursive call in $H_i$ has size at least $\alpha(H_i)-n(k'-1)/p$.  Hence, with probability at least
$\tfrac{1}{3n}\cdot a(G',k'-1)=a(G',k')$, we return 
an independent set of size at least $\alpha(C_i)+\alpha(H_i)-n(k'-1)/p=\alpha(G'-N(C_i))-n(k'-1)/p\ge \alpha(G')-n/p-n(k'-1)/p=\alpha(G')-nk'/p$.
\item If (c) holds, then with probability at least $\tfrac{k'q\alpha(G')}{3n}\ge \tfrac{(k')^2q}{3p}\ge \tfrac{q}{3p}$,
the algorithm takes the branch (c) and chooses $u\in V(G')$ such that $\alpha(G'-N(u)) = \alpha(G')$.
We return the independent set obtained by the recursive call on $G'-N(u)$, which has size
at least $(1-k'/p)\alpha(G'-N(u))=(1-k'/p)\alpha(G')$ with probability at least $a(G'-N(u),k')$ by the induction hypothesis.
Note that $d(G'-N(u))\le d(G')-1$, and thus the probability we return an independent set of size at least $(1-k'/p)\alpha(G')$
is at least $\bigl(\tfrac{q}{3p}\bigr)\cdot a(G'-N(u),k')\ge a(G',k')$.
\end{itemize}
\end{proof}

We now improve the probability by iterating.

\begin{theorem}\label{thm-main}
There exists a randomized algorithm that, for input integers
$k\ge 0$ and $t\ge 1$ and an $n$-vertex graph $G$ of induced odd cycle packing number at most $k$, returns in
time $O_{k,t}(n^{k+4})$ an independent set of $G$ whose size is at least $\alpha(G)-n/t$ with probability at least $1-1/e$.
\end{theorem}
\begin{proof}
Let $q=\Omega_{k,t}(n^{-k})$ be the lower bound on the probability that the algorithm from Lemma~\ref{lemma-main} succeeds.
Run the algorithm $\lceil 1/q\rceil$ times (with independent random choices), and return the largest of the obtained independent sets.
The probability that none of them has size at least $\alpha(G)-n/t$ is at most $(1-q)^{1/q}\le e^{-1}$,
and thus the algorithm succeeds with probability at least $1-1/e$.
\end{proof}

Of course, we can further iterate the algorithm $a$ times, reducing the probability of a result worse than $\alpha(G)-n/t$ to
at most $e^{-a}$.

\section{$\chi$-boundedness}\label{secChiBound}

In this section we show that the classes of graphs with bounded induced odd packing number are $\chi$-bounded, that is, their chromatic number is bounded
by a function of their maximum clique size.  Let us start with the triangle-free case, where we need to show an absolute bound
on the chromatic number.
\begin{lemma}\label{lemma-trfree}
Every triangle-free graph $G$ satisfies $\chi(G)\le 2+5\iocp(G)$.  Furthermore, if $G$ has odd girth at least $7$,
then $\chi(G)\le 2+4\iocp(G)$, and if $G$ has girth at least $7$, then $\chi(G)\le 3+\iocp(G)$.
\end{lemma}
\begin{proof}
We prove the claim by induction on the induced odd cycle packing number.  If $\iocp(G)=0$, then $G$ is bipartite and $\chi(G)\le 2$,
hence suppose that $\iocp(G)>0$ and the claim holds for all graphs with smaller induced odd cycle packing number.  Let $C$ be a shortest
odd cycle in $G$, which is necessarily induced.  Since $\iocp(G-N[V(C)])\le \iocp(G)-1$, we can color
$G-N[V(C)]$ by the induction hypothesis, and it suffices to show how to color $G[N[V(C)]]$ using at most $5$ extra colors,
or at most $4$ extra colors if $G$ has odd girth at least $7$, or using at most one extra color if $G$ has girth at least $7$.

Let $A=\{z_1,z_2,z_3\}$ be a set consisting of three consecutive vertices of $C$ and let $R=N[A]$.
By Lemma~\ref{lemma-short-nbr}, $G[N[V(C)]\setminus R]$ is bipartite.  Since $G$ is triangle-free, the neighborhood of any vertex is
an independent set, and $N(\{z_1,z_3\})$ is disjoint from $N(z_2)$.  
Hence, we can use two new colors to color $G[N[V(C)]\setminus R]$, one new color for $N(z_2)$, one new color for $N(z_1)$, and
one new color for $N(z_3)\setminus N(z_1)$, using five extra colors in total.

Moreover, if $G$ has odd girth at least $7$, then $N(\{z_1,z_3\})$ is an
independent set as well; hence, we can use two new colors to color
$G[N[V(C)]\setminus R]$, one new color for $N(z_2)$ (which includes $z_1$ and
$z_3$) and one new color for $N(\{z_1,z_3\})$ (which includes $z_2$), using
only four extra colors in total.

In the case that $G$ has girth at least $7$, we claim that $N(V(C))$ is an independent set.
Indeed, suppose for a contradiction 
vertices $w,z\in N(C)$ are adjacent, and let $w'$ and $z'$ be neighbors of $w$ and $z$ in $C$, respectively.
Since $G$ has girth at least $7$, the distance between $w'$ and $z'$ in $C$
is greater than three.  
However, then $C+w'wzz'$ contains an odd cycle shorter than $C$, which is a contradiction.
Hence, we can use three of the colors used on $G-N[V(C)]$ to color $C$ and one extra color
to color $N(V(C))$.  Let us remark that we indeed obtain an upper bound of $\chi(G)\le 3+\iocp(G)$ rather than
$\chi(G)\le 2+\iocp(G)$: While the latter bound works for $\iocp(G)=0$, we need an extra color
when $\iocp(G)=1$.
\end{proof}

For the case when $G$ contains triangles, let $f(0,\omega)=2$ for every positive integer $\omega$, and for $k\ge 1$,
let us inductively define $f(k,\omega)=\omega+(2+5k)\binom{\omega}{2}+f(k-1,\omega)\binom{\omega}{3}$.
\begin{theorem}\label{thm-chi}
Every graph $G$ satisfies $\chi(G)\le f(\iocp(G),\omega(G))$.
\end{theorem}
\begin{proof}
We prove the claim by induction on the induced odd cycle packing number.  If $\iocp(G)=0$, then $G$ is bipartite and $\chi(G)\le 2$,
hence suppose that $\iocp(G)>0$ and the claim holds for all graphs with smaller induced odd cycle packing number.  Let $K$ be a largest clique
in $G$, and for each $v\in V(G)\setminus K$, let $A(v)$ denote the set of vertices of $K$ to which $v$ is not adjacent; the maximality of $K$
implies $A(v)\neq\emptyset$.  Let $A'(v)$ be an arbitrary subset of $A(v)$ of size $\min(3,|A(v)|)$.  For a set $X\subseteq K$ with
$|X|\in\{1,2,3\}$, let $B(X)=\{v\in V(G)\setminus K:A'(v)=X\}$.  If $|X|=1$, then the maximality of $K$ implies that $B(X)$ is an independent
set; for each $1$-element set $X$, we use one color for all vertices of $X\cup B(X)$.  If $|X|=2$, then the maximality of $K$ implies
$G[B(X)]$ is triangle-free, and by Lemma~\ref{lemma-trfree}, we can use $2+5\iocp(G)$ colors to color $G[B(X)]$.
Finally, if $|X|=3$, then $\iocp(G[B(X)])\le\iocp(G[B(X)])-1$, since all vertices in $B(X)$ are non-adjacent to the triangle
induced by $X$; hence, we can use $f(\iocp(G)-1,\omega(G))$ colors to color $G[B(X)]$.  Summing the numbers of colors over all
choices of $X$, we conclude that at most $f(\iocp(G),\omega(G))$ colors are used to color $G$.
\end{proof}
Let us remark that we can obtain the coloring as in Theorem~\ref{thm-chi} in polynomial time: instead of choosing $K$ as a largest
clique, the inspection of the proof shows that it suffices to choose one which cannot be enlarged by adding at most three and removing
at most two vertices, and such a clique can be found in polynomial time.

\section{Lower-bound for $\chi$-bounding function}\label{secLowerBound}

In the previous section we showed that every graph $G$ has chromatic number at most $f(\iocp(G),\omega(G))$, where $f$ is of order
roughly $\omega(G)^{3\iocp(G)}$. We do not know whether this upper bound is tight. In this section we show that there exist graphs with induced odd cycle packing number one
whose chromatic number is almost quadratic in $\omega(G)$.
We give a probabilistic construction. In order to carry out the probabilistic calculations, we will use the following bounds. 

\begin{lemma}[See~\cite{molloy2013graph}, Chernoff Bound]\label{chernoff}
Suppose $X$ is a sum of independent $\{0,1\}$-variables. For any $\delta > 0$,
$$\Prob \left[ X \leq (1-\delta)\E[X] \right] \leq e^{-\frac{\delta^2\E[X]}{2}}.$$
\end{lemma}

\begin{lemma}[See~\cite{molloy2013graph}, Talagrand's Inequality II]\label{talagrand}
Let $X$ be a non-negative random variable, not identically $0$,
which is determined by $n$ independent trials $T_1$, \ldots, $T_n$, and satisfying the following for some integers $c, r > 0$:
\begin{itemize}
\item Changing the outcome of any one trial can change $X$ by at most $c$.
\item For any non-negative integer $d$, if $X \geq d$, then there is a set of at most $rd$ trials whose outcomes certify that $X \geq d$.
\end{itemize}
Then for every non-negative $t \leq \E[X]$,
$$\Prob \left[ |X - \E[X]| > t + 60c \sqrt{r\E[X]} \right] \leq 4e^{-\frac{t^2}{8c^2r\E[X]}}.$$
\end{lemma}

We are now ready to give the construction, which is a variation on a standard construction of triangle-free graphs with no large independent set.

\begin{theorem}\label{thmProbConst}
There exists a family of graphs with induced odd cycle packing number at most one and with arbitrarily large clique number such that every graph $H$
in this family satisfies $\chi(H) = \Omega\Bigl(\frac{\omega^2(H)}{\log^2 \omega(H)}\Bigr)$.
\end{theorem}
\begin{proof}
Let $G$ be a random $G(n,p)$ graph for $p = 1/k$, where $k$ is a sufficiently large even integer and $n = k^2/2$.

Suppose $A_1,A_2\subset V(G)$ are disjoint and have size three, and $G$ contains all nine edges
with one end in $A_1$ and the other end in $A_2$; in this case, we say that the set of these nine edges
\emph{forms a $K_{3,3}$}. We say a set of edges of $G$ is \emph{bad} if it can be partitioned into subsets of size $9$, each of which forms a $K_{3,3}$.
Finally we say a set of edges is \emph{subbad} if it is a subset of a bad set. 

We construct a graph $G_0$ by deleting a maximal bad set $B$ from $G$.  By the maximality of $B$, no set of edges of $G_0$ forms a $K_{3,3}$.
Let $H$ be the complement of $G_0$. Consider any disjoint (odd) cycles in $H$. 
If there were no edge between these cycles, then any pair of triples of vertices taken one from each cycle would induce a complement of a supergraph of $K_{3,3}$ in $G_0$,
which is a contradiction.  Consequently $\iocp(H) \leq 1$.  Furthermore, an analogous argument shows that $\alpha(H)\le 5$,
and thus $\chi(H) \geq n/5 = \Omega(k^2)$.

Therefore, it suffices to argue that $\omega(H) =\alpha(G_0)=O(k \log k)$ with non-zero probability.
Let us consider a set $S \subseteq V(G)$ of size at least $k$ and define the following random variables.
\begin{align*}
X_1 &= | E(G[S]) |\\
X_2 &= \max \bigl\{ |Z| : \text{$Z \subseteq E(G[S])$, $Z$ is subbad in $G$}\bigr\}
\end{align*}
The probability that $S$ is an independent set in $G_0$ is at most $\Prob[X_1 \leq X_2]$.
Indeed, if $S$ is independent, then $E(G[S])\subseteq B$ is subbad, and thus $X_1=X_2$.
Let $s=\binom{|S|}{2}=\Omega(k^2)$, so that $\E[X_1]=s/k$. Using the Chernoff bound (Lemma~\ref{chernoff}), we obtain
$$\Prob\left[X_1 \leq\frac{s}{2k}\right]=\Prob[X_1 \leq \E[X_1]/2] \leq e^{-\frac{\E[X_1]}{8}} = e^{-\frac{s}{8k}}.$$
Let $Z_2=\{e\in E(G[S]): \text{some set containing $e$ forms a $K_{3,3}$ in $G$}\}$.
Clearly, if $Z\subseteq E(G[S])$ is subbad, then $Z\subseteq Z_2$, and thus $X_2\le |Z_2|$.
For distinct vertices $x,y\in V(G)$, the probability that $xy$ is an edge and some set containing $xy$ forms a $K_{3,3}$ in $G$
is at most $\tfrac{n^4}{k^9}=\tfrac{1}{16k}$, and thus $\E[X_2]\le \E[|Z_2|]\le \tfrac{s}{16k}$.
Let $\delta=\tfrac{s}{16k}-\E[X_2]\ge 0$ and let $X'_2=X_2+\delta$, so that
$\E[X'_2]=\tfrac{s}{16k}$.  Note that flipping the existence of a single edge in $G$
changes $X'_2$ by at most $9$.
Furthermore, when $X'_2 \geq d$, there exist at most $9d$ edges in $G$ whose presence certifies this is the case.
Hence, we can apply the Talagrand's inequality (Lemma~\ref{talagrand}) with $c=r=9$
for $t=\E[X'_2]=\tfrac{s}{16k}$.  Note that $60c\sqrt{r\E[X'_2]}\le \E[X'_2]$ for $k$ large enough, since $s=\Omega(k^2)$.  Hence, we have
$$\Prob\left[X'_2>\frac{3s}{16k}\right]=\Prob[X'_2 > 3\E[X'_2]] \leq 4e^{-\frac{\E[X'_2]}{5832}}=4e^{-\frac{s}{93312k}}.$$
It follows that
\begin{align*}
\Prob[X_1\le X_2] &\leq \Prob[X_1\le X'_2]\le \Prob\left[X_1\le \frac{s}{2k}\right]+\Prob\left[X'_2> \frac{3s}{16k}\right]\\
&\leq e^{-\frac{s}{8k}} + 4e^{-\frac{s}{93312k}} < e^{-\frac{|S|^2}{200000k}}
\end{align*}
for $k$ large enough.

Therefore, for any set $S\subseteq V(G_0)$ of size $q\ge k$, the probability that $S$ is an independent set
in $G_0$ is at most $e^{-\frac{q^2}{200000k}}$.  Hence, the probability that $G_0$ contains an independent set of
size $q$ is less than
\begin{align*}
\binom{n}{q}e^{-\frac{q^2}{200000k}}&\le \left(\frac{ne}{q}\right)^qe^{-\frac{q^2}{200000k}}\\
&\le (ke)^qe^{-\frac{q^2}{200000k}}=\exp\Bigl(q\log(ke)-\tfrac{q^2}{200000k}\Bigr),
\end{align*}
that is, smaller than $1$ when $q\ge 200000k\log(ke)$.
\end{proof}

As noted in the proof, the probabilistic construction used is unnecessarily
restrictive, excluding all disjoint cycles regardless of parity. Similarly, all
cycles or paths of length $\geq 8$ are excluded. From the point of view of
vertex 6-tuples, all of these structures exhibit very similar properties. It
would seem that achieving a distinction between these patterns and the ones
that are necessary to avoid requires more global conditions and thus a much
more refined approach. 

\section{QPTAS assuming only bounded $\iocp$}\label{secQPTAS}

The fact that triangle-free graphs with bounded $\iocp$ have bounded chromatic number has the following easy consequence.

\begin{lemma}\label{lemma-degree}
For all integers $p\ge 1$, every graph $G$ with $n$ vertices satisfies at least one of the following conditions:
\begin{itemize}
\item $G$ has an independent set of size at least $\tfrac{n}{4+10\iocp(G)}$, or
\item every maximal packing of triangles in $G$ contains a triangle $T$ such that $\alpha(G-N(T))\ge (1-1/p)\alpha(G)$, or
\item $G$ contains a vertex $v$ of degree at least $\tfrac{n}{18p}$ such that $\alpha(G-N(v))=\alpha(G)$.
\end{itemize}
\end{lemma}
\begin{proof}
Let $\XX=\{T_1,\ldots, T_m\}$ be a maximal packing of triangles in $G$.  The graph $G-V(\XX)$ is triangle-free,
and by Lemma~\ref{lemma-trfree}, $\chi(G-V(\XX))\le 2+5\iocp(G)$.  Consequently, $\alpha(G)\ge \alpha(G-V(\XX))\ge \frac{n-3m}{2+5\iocp(G)}$.
Suppose that $G$ does not have any independent set of size at least $\tfrac{n}{4+10\iocp(G)}$, and thus $m\ge n/6$.
Let $J$ be a largest independent set in $G$.  If $|N(T_i)\cap J|\le |J|/p$ for some $i\in\{1,\ldots,m\}$, then
the second outcome of the lemma holds.

Otherwise, for each $i\in \{1,\ldots,m\}$, there exists $v_i\in V(T_i)$ such that $|N(v_i)\cap J|>\tfrac{|J|}{3p}$.
Consequently, there exists a vertex $v\in J$ such that $|N(v)\cap \{v_1,\ldots,v_m\}|\ge \tfrac{m}{3p}$,
and thus $\deg v\ge \tfrac{m}{3p}\ge \tfrac{n}{18p}$.  Since $v\in J$, we have $\alpha(G-N(v))=\alpha(G)$.
\end{proof}

Combining this lemma with Theorem~\ref{thm-main}, we obtain a QPTAS for the maximum independent set
in graphs of bounded induced odd cycle packing number.

\begin{theorem}\label{thm-qptas}
There exists a randomized algorithm that, for input integers $k\ge 0$ and $p\ge 1$ and an $n$-vertex graph $G$ of induced odd cycle packing number at most $k$, returns in
time $n^{O(k+p\log n)}$ an independent set of $G$ whose size is at least $(1-k/p)\alpha(G)$ with probability at least $1/2$.
\end{theorem}
\begin{proof}
If $k=0$ (so $G$ is bipartite), we return the largest maximum independent set obtained
via a maximum flow algorithm.  Otherwise, we find a maximal packing of triangles $\XX$ in $G$ greedily, and return the largest
of the independent sets obtained by
\begin{itemize}
\item[(a)] running the algorithm from Theorem~\ref{thm-main} $n$ times with $t=(4+10k)p$,
\item[(b)] for each $T\in\XX$, running the algorithm recursively for $G-N[T]$ with $k$ replaced by $k-1$ and adding a vertex of $T$ to the returned independent set, and
\item[(c)] for each $v\in V(G)$ of degree at least $\tfrac{n}{18p}$, running the algorithm recursively for $G-N(v)$.
\end{itemize}
Each recursive call either decreases $k$ or decreases the number of vertices by a factor of at most $\bigl(1-\tfrac{1}{18p}\bigr)$,
implying the total number of the calls of the procedure is at most $n^{O(k+p\log n)}$.  Iterating the algorithm from Theorem~\ref{thm-main} $n$ times
for an induced subgraph $G'$ of $G$ ensures we fail to find an independent set of size at least $\alpha(G')-n/t$ with probability at most $2^{-n}$.
Hence, with probability at least $1-n^{O(k+p\log n)}2^{-n}>1/2$ (for $n$ large enough---for small $n$, we can just find the largest independent set by brute force), we can assume that throughout the run of the algorithm,
in part (a) for an induced subgraph $G'$ of $G$, at least one of the returned independent sets has size at least
$\alpha(G')-n/t$.

If $\alpha(G)\ge \tfrac{n}{4+10k}$, then in (a) we return an independent set of size
at least $\alpha(G)-n/t\ge \alpha(G)-(4+10k)\alpha(G)/t=(1-1/p)\alpha(G)$.  If $G$ contains a vertex $v$ of degree at least
$\tfrac{n}{18p}$ such that $\alpha(G-N(v))=\alpha(G)$, then in (c) the corresponding recursive call gives an independent set of size
at least $(1-k/p)\alpha(G-N(v))=(1-k/p)\alpha(G)$.

If neither of these conditions holds, Lemma~\ref{lemma-degree} implies 
there exists a triangle $T\in\XX$ such that $\alpha(G-N(T))\ge (1-1/p)\alpha(G)$.
The recursive call in (b) returns an independent set $I$ of $G-N[T]$ of size at least $(1-(k-1)/p)\alpha(G-N[T])$,
and since $\alpha(G-N(T))=\alpha(G-N[T])+1$, the addition of a vertex of $T$ turns $I$ into an independent set of size at least
$(1-(k-1)/p)(\alpha(G-N(T))-1)+1\ge (1-(k-1)/p)\alpha(G-N(T))\ge (1-(k-1)/p)(1-1/p)\alpha(G)\ge (1-k/p)\alpha(G)$.
Hence, the algorithm is correct.
\end{proof}

Considering Theorem~\ref{thm-qptas}, it is of course natural to ask whether the maximum independent set problem
admits a PTAS on graphs with bounded induced odd cycle packing number without any other assumptions.

\section*{Acknowledgments}

We would like to thank the anonymous reviewer, whose suggestion greatly simplified our analysis of
the algorithm in Theorem~\ref{thm-main}.


\begin{thebibliography}{1}

\bibitem{bock}
{\sc A.~Bock, Y.~Faenza, C.~Moldenhauer and A.~J.~Ruiz-Vargas},
  {\em Solving the Stable Set Problem in Terms of the Odd Cycle Packing Number}.
  in 34th International Conference on Foundation of Software Technology and Theoretical Computer Science (FSTTCS 2014),
  pp.~187–198.

\bibitem{indoc}
{\sc M.~Bonamy, {\'E}.~Bonnet, N.~Bousquet, P.~Charbit, P.~Giannopoulos, E.J.~Kim, P.~Rza\.{z}ewski, F.~Sikora and S.~Thomass{\'{e}}},
  {\em {EPTAS} and Subexponential Algorithm for Maximum Clique on Disk and Unit Ball Graphs}, J. ACM Volume 68, Issue 2, March 2021.

\bibitem{BCT}
{\sc M.~Bonamy, P.~Charbit, and S.~Thomass\'{e}}, {\em Graphs with large
  chromatic number induce $3k$-cycles}, arXiv, 1408.2172 (2014).

\bibitem{BGKRS}
{\sc {\'E}.~Bonnet, P.~Giannopoulos, E.~J. Kim, P.~Rz\c{a}\.{z}ewski, and F.~Sikora},
  {\em {QPTAS and Subexponential Algorithm for Maximum Clique on Disk Graphs}},
  in 34th International Symposium on Computational Geometry (SoCG 2018),
  B.~Speckmann and C.~D. T{\'o}th, eds., vol.~99 of Leibniz International
  Proceedings in Informatics (LIPIcs), Dagstuhl, Germany, 2018, Schloss
  Dagstuhl--Leibniz-Zentrum fuer Informatik, pp.~12:1--12:15.

\bibitem{sperft}
{\sc M.~Chudnovsky, N.~Robertson, P.~Seymour, and R.~Thomas}, {\em The strong
  perfect graph theorem}, Annals of Mathematics, 164 (2006), pp.~51--229.

\bibitem{GLS}
{\sc M.~Gr\"{o}tschel, L.~Lov\'{a}sz, and A.~Schrijver}, {\em Polynomial
  algorithms for perfect graphs}, North-Holland mathematics studies, 88 (1984),
  pp.~325--356.

\bibitem{gyarfas1987}
{\sc A.~Gy{\'a}rf{\'a}s}, {\em Problems from the world surrounding perfect
  graphs}, Applicationes Mathematicae, 19 (1987), pp.~413--441.

\bibitem{molloy2013graph}
{\sc M.~Molloy and B.~Reed}, {\em Graph colouring and the probabilistic
  method}, vol.~23, Springer Science \& Business Media, 2013.

\bibitem{SS}
{\sc A.~Scott and P.~Seymour}, {\em Induced subgraphs of graphs with large chromatic number. I. Odd holes},
  Journal of Combinatorial Theory, Series B, 121 (2016), pp.~ 68--84.

\bibitem{colnonap}
{\sc D.~Zuckerman}, {\em Linear degree extractors and the inapproximability of
  {M}ax {C}lique and {C}hromatic {N}umber}, Theory of Computing, 3 (2007),
  pp.~103--128.

\end{thebibliography}

\end{document}